%% file: graphe.tex
\documentclass{article}

\usepackage{etex}
\usepackage{geometry}
\usepackage{makeidx}
\makeindex{}

\input{preambule.tex}

\input{en-preambule.tex}

\newcommand{\type}{\emph}
\newcommand{\param}{\textbf}
\newcommand{\var}{\texttt}
\begin{document}
\date{} \title{(Quasi-)linear time algorithm to compute LexDFS, LexUP
  and LexDown orderings} \author{Arthur Milchior}

\maketitle
\begin{abstract}
  We consider the three graph search algorithm LexDFS, LexUP and
  LexDOWN. We show that  LexUP orderings can be computed in
  linear time by an algorithm similar to the one which compute
  LexBFS. Furthermore, LexDOWN orderings and LexDFS orderings can
  be computed in time $\left(n+m\log m\right)$ where $n$ is the number
  of vertices and $m$ the number of edges.
\end{abstract}
\section{Introduction}
A graph search is a mechanism for systematically visiting the vertices
of a graph. Deep-First Search (DFS) and Breadth-First Search (BFS)
have been studied for decades (see e.g. \cite{Cormen}).  Those two
graph searches can be computed in linear time. A particular kind of
BFS, the Lexicographical BFS (LexBFS), has then been introduced in
\cite{DBLP:journals/siamcomp/RoseTL76}. And by similarity, the
Lexicographical DFS (LexDFS) has been studied in
\cite{DBLP:journals/siamdm/CorneilK08}. And then LexUP and LexDOWN in
\cite{dusart:tel-01273352}.

A LexBFS ordering of a graph $G$ is a possible output of a LexBFS
search applied to $G$. While the LexBFS al algorithm runs in time
$\bigO{nm}$ where $n$ is the number of vertices and $m$ the number of
edges, a LexBFS ordering can be computed in time $\bigO{n+m}$. LexDFS,
LexUP and LexDOWN also run in time $\bigO{nm}$. We show that a LexUP
ordering can be computed in linear time by an algorithm similar to the
one which compute a LexBFS ordering. Furthermore, we prove that a
LexDOWN ordering and a LexDFS ordering can be computed in time
$\bigO{n+m \log m}$.

Definitions are given in \autoref{sec:def}. The four graph search
algorithms considered in this paper are given in
\autoref{sec:search}. An efficient algorithm to compute LexDFS and
LexDOWN ordering are given in \autoref{sec:DFS-DOWN}. Finally,
efficient algorithm is given to compute a LexUP ordering in
\autoref{sec:BFS-UP}.
\section{Definition}\label{sec:def}
Definitions used in this paper are now introduced. Most of those
definitions are standard.  Let $\N$\index{N@$\N$} be the set of
non-negative integer.  For $A$ a finite set, $\card
A$\index{A@$\card A$} denotes the cardinality\index{Cardinality} of
$A$. 

A word\index{Word} on $\N$ is a sequence $a_{1}\dots a_{n}$, with
$a_{i}\in \N$. The empty word\index{Empty word} is denoted
$\epsilon$\index{epsilon@$\epsilon$}. For $\mathbf a=a_{1}\dots a_{n}$
and $\mathbf b=b_{1}\dots b_{m}$ two words over $A$, it is said that
$\mathbf a$ is (lexicographically) smaller\index{Lexicographically
  smaller} than $\mathbf b$ if there exists $i\le \min(n,m)$ such
that, for all $1\le j\le i$, $a_{j}=b_{j}$, and (either $i=n<m$ or
$a_{i+1}<b_{i+1}$).

\subsection{Graph}
A (undirected) graph $G$ with a source\index{(undirected) graph with a
  source} is a 3-tuple $(V,E,s)$\index{$(V,E,s)$} where $V$ is a
finite set, $E$ is a set of subsets of $V$ whose elements's
cardinality is 2 and $s\in V$.  The elements of $V$ are called
vertices\index{Vertices}. The elements of $E$ are called
edges\index{Edges}.  The vertex $s$ is called the source\index{Source
  of a graph}.

A vertex $v$ is said to be a neighbor\index{neighbor} of $w$ if
$\set{v,w}\in E$. The neighborhood of a vertex\index{Neighborhood of a
  vertex} $v$ is the set of neighbor of $w$, it is denoted
$N(v)$\index{N@$N(v)$}. Formally, $N(v)=\set{w\mid \set{v,w}\in E}$. The
degree of $v$, denoted $d(v)$\index{dv@$d(v)$}, is the cardinality of
its neighbourhood. Formally, $d(v)=\card{N(v)}$. 

Two vertices $v,w\in V$ are said to be connected\index{Connected
  vertices} if there exists a sequence $v=v_{0},\dots, v_{p}=w$ such
that, for all $0\le i<p$, $\set{v_{i},v_{i+1}}\in E$. A graph is said
to be connex\index{Connex graph} if all pair of distinct vertices are
connected. 



\subsection{Data structures}
In this section, we list the data structures used in this paper. We
list the operation those data structures admit, and their time
complexity. All of those notions are standard (see
e.g. \cite{Cormen}).

In this paper, each type is represented as \type{type}, each variable
is represented as \var{var} and each function of parameter of an
object $o$ is represented as $o$.\param{param}.

It is assumed thourought this paper that \type{integer}s can be incremented
and compared in constant time. During execution of the algorithm of
this paper of a graph $(V,E)$, all \type{integer} variables are interpreted
by a number whose absolute value is at most $\max(\card V,2\card
E)$. Hence, the constant time assumption is relatively safe. Assing a
value $x$ to a variable $v$ is denoted $v:=x$ and is assumed to take
constant time.
\paragraph{Arrays}

It is assumed in this paper that \type{array}s are created in time linear to
their numbers of elements.  The elements of an \type{array} $A$ with $n$
elements are numbered from $1$ to $n$.  The $i$-th element of $A$ is
denoted $A[i]$, and can be read and assigned in constant time.

\paragraph{Doubly linked lists}

In this paper, all lists are assumed to be \type{doubly-linked
  list}s. A \type{doubly-linked list} of elements of type $t$ is a sequence
of \type{node}s, with direct access to its first and last \type{node}s. Each \type{node}
contains a value of type $t$. Each \type{node} has also a direct access to
its list, to the preceding and following \type{node}s.  A \type{doubly-linked list}
$l$ admits the following constant-time operations:
\begin{itemize}
\item Access to its first \type{node}: $l.\param{first}$.
\item Access to its last \type{node}: $l.\param{last}$.
\item Adding a \type{node} $c$ to the head of $l$: $l.\param{add-first}(c)$.
\item Adding a \type{node} $c$ to the end of $l$: $l.\param{add-last}(c)$.
\end{itemize}
A list with $n$ \type{node}s can be sorted in time $\bigO{n.\log(n)}$,
assuming that the comparison of two \type{node}s of the list can be done
in constant time: $l.\param{sort}$. The order will always be clear in
the algorithms of this paper.

A \type{node} $e$ of a \type{doubly-linked list} $l$ admits the following
operations:
\begin{itemize}
\item access to the preceding \type{node}: $e.\param{pred}$,
\item access to the following \type{node}: $e.\param{next}$,
\item access to the value at position $e$: $e.\param{value}$,
\item inserting a value $i$ of type $t$ in a new \type{node} after $e$: $e.\param{add-after}(i)$ and
\item inserting a value $i$ of type $t$ in a new \type{node} before $e$:
  $e.\param{add-before}(i)$ and
\item removing $e$: $e.\param{remove}$.
\end{itemize}
Note that the first value of type $t$ of a \type{list} $l$ is
$l.\param{first}.\param{value}$ and not $l.\param{first}$. Indeed,
$l.\param{first}$ is a \type{node} and not a value of type $t$.

\paragraph{Graphs}
A \type{graph} $G$ is represented as an \type{array} of size $n$. The $i$-th element
of the \type{array} contains the list of neighbors of $v_{i}$.  Formally,
$N(i)$ should be represented as $G[i]$, however, $N(i)$ is used in the
algorithms of this paper for the sake of the readability.

\section{Graph search algorithm}\label{sec:search}
In this section, the four graph search algorithms considered in this
paper are considered. A graph search algorithm is an algorithm as in
\autoref{def:search}. Note that the standard definition of
graph search algorithms is more general than the one used in this
paper.
\begin{algorithm}[!h]
  \caption{Definition of a graph search algorithm}\label{def:search}
  \KwIn{An undirected \type{graph} $G=(V,E,s)$ with \var{n} vertices} \KwOut{an
    ordering $\sigma$ of the vertices of $G$}
  \lnl{search:epsilon}assign .\param{label} $\epsilon$ to all vertices\;
  \lnl{search:infty-s}assign .\param{label} $[\infty]$ to $s$\;
  \lnl{search:loop}\ForEach{\var{i} from $1$ to \var{n}}{
    \lnl{search:choose}pic an \var{unnumbered} \type{vertex} \var{vertex} with lexicographically maximal .\param{label}\;
    \lnl{search:choose}$\sigma$[\var{i}]:=\var{vertex}\;
    \lnl{search-def-loop-neighb}\ForEach{\var{unnumbered} vertex  \var{\var{neighb}}$\in$ N(\var{vertex})}{
      \lnl{search:update}update \var{\var{neighb}}.\param{label}\; }
  }
  Output $\sigma$\;
\end{algorithm}
The only difference beween the four graph search algorithms considered
in this paper appears in Line \ref{search-def-loop-neighb} of
\autoref{def:search}.  

In this paper, the label is always a \type{list} of \type{integer}s. Each update
always takes constant time and add exactly an \type{integer} to the label.
The time complexity of this algorithm is now considered.
\begin{lemma}
  Let $G=(V,E,s)$ a graph with $n$ vertices and $m$ edges.  Assuming
  the update consists in adding an \type{integer} in the front or in
  the rear of the list, the time complexity of \autoref{def:search} is
  $\bigO{nm}$.
\end{lemma}
\begin{proof}
  Let us first consider the labels. The label of a vertex $v$ contains
  at most $\card{N(v)}$ elements. Hence the sum of the length of the
  label is at most $2m$.
  
  Lines \ref{search:epsilon}, \ref{search:infty-s} are executed once
  and in constant time. Hence their time cost is $\bigO{1}$.  Each
  execution of Line \ref{search:choose} may have to read the entire
  labels of each vertex. Finding the maximal label then cost
  $\bigO{m}$-times. Since this line is executed $n$ times, this Line
  costs $\bigO{nm}$-times.  Line \ref{search:choose} can be executed
  in constant time, and is executed $n$ times, hence it costs
  $\bigO{n}$ time. Finally, each execution of line \ref{search:update}
  takes constant time. And this line is executed once for each
  $1\le i\le n$ and $n\in N(v_{i})$, hence it is executed
  $2m$-times. Thus, it costs $\bigO{m}$ time.

  Finally, the whole algorithm runs in time $\bigO{nm}$.
\end{proof}
For $\mathcal A$ a graph search algorithm, an $\mathcal A$
ordering\index{A ordering@$\mathcal A$ ordering} of $G$ is a possible
output of $\mathcal A$ on $G$.
\tikzstyle{vertex}=[circle,draw,color=black]

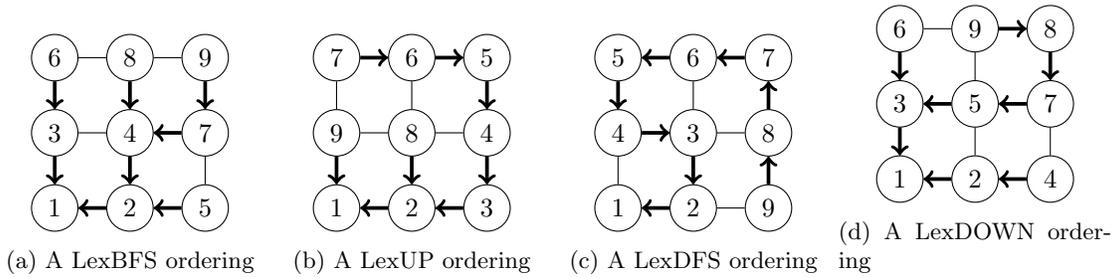
\begin{figure}
  \begin{subfigure}[b]{.24\textwidth}
    \centering
    \begin{tikzpicture}
      \node[vertex] (00) at (0,0) {1};
      \node[vertex] (10) at (1,0) {2};
      \node[vertex] (20) at (2,0) {5};
      \node[vertex] (01) at (0,1) {3};
      \node[vertex] (11) at (1,1) {4};
      \node[vertex] (21) at (2,1) {7};
      \node[vertex] (02) at (0,2) {6};
      \node[vertex] (12) at (1,2) {8};
      \node[vertex] (22) at (2,2) {9};
      
      \path (00) edge (01) edge (10);
      \path (22) edge (21) edge (12);
      \path (02) edge (12) edge (01);
      \path (20) edge (21) edge (10);
      \path (11) edge (01) edge (10) edge (12) edge (21);
      \path[->,line width=1.4pt] (10)edge (00);
      \path[->, line width=1.4pt] (01)edge (00);
      \path[->, line width=1.4pt] (11)edge (10);
      \path[->, line width=1.4pt] (20)edge (10);
      \path[->, line width=1.4pt] (02)edge (01);
      \path[->, line width=1.4pt] (21)edge (11);
      \path[->, line width=1.4pt] (12)edge  (11);
      \path[->, line width=1.4pt] (22)edge  (21);    
    \end{tikzpicture}
    \caption{A LexBFS ordering}
    \label{fig:BFS}
  \end{subfigure}
  \begin{subfigure}[b]{.24\textwidth}
    \centering
    \begin{tikzpicture}
      \node[vertex] (00) at (0,0) {1};
      \node[vertex] (10) at (1,0) {2};
      \node[vertex] (20) at (2,0) {3};
      \node[vertex] (01) at (0,1) {9};
      \node[vertex] (11) at (1,1) {8};
      \node[vertex] (21) at (2,1) {4};
      \node[vertex] (02) at (0,2) {7};
      \node[vertex] (12) at (1,2) {6};
      \node[vertex] (22) at (2,2) {5};

      \path (00) edge (01) edge (10);
      \path (22) edge (21) edge (12);
      \path (02) edge (12) edge (01);
      \path (20) edge (21) edge (10);
      \path (11) edge (01) edge (10) edge (12) edge (21);

      \path[->, line width=1.4pt] (10)edge (00);
      \path[->, line width=1.4pt] (20) edge  (10);
      \path[->, line width=1.4pt] (21) edge  (20);
      \path[->, line width=1.4pt] (22) edge  (21);
      \path[->, line width=1.4pt] (12) edge  (22);
      \path[->, line width=1.4pt] (02) edge  (12);
      \path[->, line width=1.4pt] (11) edge  (10);
      \path[->, line width=1.4pt] (01) edge  (00);
    
    \end{tikzpicture}
    \caption{A LexUP ordering}
    \label{fig:UP}
  \end{subfigure}
  \begin{subfigure}[b]{.24\textwidth}
    \centering
    \begin{tikzpicture}
      \node[vertex] (00) at (0,0) {1};
      \node[vertex] (10) at (1,0) {2};
      \node[vertex] (20) at (2,0) {9};
      \node[vertex] (01) at (0,1) {4};
      \node[vertex] (11) at (1,1) {3};
      \node[vertex] (21) at (2,1) {8};
      \node[vertex] (02) at (0,2) {5};
      \node[vertex] (12) at (1,2) {6};
      \node[vertex] (22) at (2,2) {7};

      \path (00) edge (01) edge (10); \path (22) edge (21) edge (12);
      \path (02) edge (12) edge (01); \path (20) edge (21) edge (10);
      \path (11) edge (01) edge (10) edge (12) edge (21);

      \path[->,line width=1.4pt] (10)edge (00);
      \path[->, line width=1.4pt] (11) edge  (10);
      \path[->, line width=1.4pt] (01) edge  (11);
      \path[->, line width=1.4pt] (02) edge  (01);
      \path[->, line width=1.4pt] (12) edge  (02);
      \path[->, line width=1.4pt] (22) edge  (12);
      \path[->, line width=1.4pt] (21) edge  (22);
      \path[->, line width=1.4pt] (20) edge  (21);
    
    \end{tikzpicture}
    \caption{A LexDFS ordering}
    \label{fig:DFS}
  \end{subfigure}
%
%
  \begin{subfigure}[b]{.24\textwidth}
    \centering
    \begin{tikzpicture}
      \node[vertex] (00) at (0,0) {1};
      \node[vertex] (10) at (1,0) {2};
      \node[vertex] (20) at (2,0) {4};
      \node[vertex] (01) at (0,1) {3};
      \node[vertex] (11) at (1,1) {5};
      \node[vertex] (21) at (2,1) {7};
      \node[vertex] (02) at (0,2) {6};
      \node[vertex] (12) at (1,2) {9};
      \node[vertex] (22) at (2,2) {8};

      \path (00) edge (01) edge (10); \path (22) edge (21) edge (12);
      \path (02) edge (12) edge (01); \path (20) edge (21) edge (10);
      \path (11) edge (01) edge (10) edge (12) edge (21);
      \path[->,line width=1.4pt] (10)edge (00);
      \path[->, line width=1.4pt] (01) edge  (00);
      \path[->, line width=1.4pt] (20) edge  (10);
      \path[->, line width=1.4pt] (11) edge  (01);
      \path[->, line width=1.4pt] (02) edge  (01);
      \path[->, line width=1.4pt] (21) edge  (11);
      \path[->, line width=1.4pt] (22) edge  (21);
      \path[->, line width=1.4pt] (12) edge  (22);
    
    \end{tikzpicture}
    \caption{A LexDOWN ordering}
    \label{fig:DOWN}
  \end{subfigure}
  \caption{Orderings with source in a corner}
  \label{fig:ex}
\end{figure}
Those four algorithms are now defined, as in
\cite{dusart:tel-01273352}.
\subsection{Lexicographic Breadth-First Search}
Let \index{Lexicographic Breadth-First Search}Lexicographic
Breadth-First Search (LexBFS)\index{LexBFS} be a graph search algorithm,
as in \autoref{def:search}, where Line \ref{search-def-loop-neighb}
is: \quo{append $n-i$ to the \var{\var{neighb}}'s label}.

Intuitively, at each step, the vertex $v$ is preferred to the vertex
$v'$ if the first numbered neighboor of $v$ have been numbered earlier
than the first numbered neighboor of $v'$.  If their first neighboor
are equal, then the same comparaison is done on the second
neighbor. And so on. If $v$ and $v'$ have $i$ and $i'$ numbered
neighbors respectively, with $i'<i$, and furthermore if the $i'$ first
numbered neighbors of $v$ are exactly the first $i$ numbered neighbors
of $v'$ in the same order, then $v$ is also preferred.

\autoref{fig:BFS} shows examples of LexBFS ordering.
Each arrow associates to a vertex $v$ its earliest numbered neigbhor.
Table \ref{tab:BFS} associate to each vertex $v$ its list of numbered
neighbors when $v$ was numbered. This table also associate to $v$ its
label.

\begin{table}[h]
  \centering
  \begin{tabular}{l||l|l|l|l|l|l|l|l}
    Vertex number &2&3&4&5&6&7&8&9\\
    \hline
    $v$'s  label in Figure \ref{fig:BFS}         &8& 8&76& 7& 6& 54& 53&21\\
    When $v$  is numbered,  it's numbered neighbors are&1& 1&23& 2& 3& 45& 46& 78\\
\end{tabular}
\caption{Label during the LexBFS search of Figure \ref{fig:BFS}}\label{tab:BFS}
\end{table}

\subsection{Lexicographic UP}
Let \index{Lexicographic UP}Lexicographic
UP (LexUp)\index{LexUP} be a graph search algorithm,
as in \autoref{def:search}, where Line \ref{search-def-loop-neighb}
is: \quo{append $i$ to the \var{\var{neighb}}'s label}.

Intuitively, at each step, the vertex $v$ is preferred to the vertex
$v'$ if the first numbered neighboor of $v$ have been numbered later
than the first numbered neighboor of $v'$.  If their first neighboor
are equal, then the same comparaison is done on the second vertex. And
so on. If $v$ and $v'$ have $i$ and $i'$ numbered neighbors
respectively, with $i'<i$, and furthermore if the $i'$ first numbered
neighbors of $v$ are exactly the first $i$ numbered neighbors of $v'$
in the same order, then $v$ is also preferred.

Note that this intuition is the same than for LexBFS, apart that the
word \quo{earlier} have been replaced by the word \quo{later}. It is
because, in both cases, \type{integer}s are prepended to the label. But in
the former case, the sequence of prepended numbers decrease while in
the second case it increases. Thus, the maximal numbers are added
earlier in LexBFS and later in LexUP.

Figure \ref{fig:UP} show an example of a LexBFS ordering.  Each arrow
associates to a vertex $v$ its first numbered neigbhor.  Table
\ref{tab:UP} associates to each vertex $v$ its label when $v$ was
numbered in each of those 4 examples respectively. Note that this list
is also its list of numbered neighbors.
\begin{table}[h]
  \centering
  \begin{tabular}{l||l|l|l|l|l|l|l|l}
    Vertex number &2&3&4&5&6&7&8&9\\
    \hline
    $v$'s label in Figure \ref{fig:UP} &1&2 &3& 4&5&6 &246 &187 \\
  \end{tabular}
\caption{Label during the LexUP search of Figure \ref{fig:ex}}
\label{tab:UP}
\end{table}

\subsection{Lexicographic Depth-First Search}
Let \index{Lexicographic Depth-First Search}Lexicographic Depth-First
Search (LexDFS)\index{LexBFS} be a graph search algorithm, as in
\autoref{def:search}, where Line \ref{search-def-loop-neighb} is:
\quo{prepend $i$ to the \var{\var{neighb}}'s label}.

Intuitively, at each step, the vertex $v$ is preferred to the vertex
$v'$ if the last numbered neighboor of $v$ have been numbered later
than the last numbered neighboor of $v'$.  If their last neighboor
are equal, then the same comparaison is done on the second last
neighbor. And so on. If $v$ and $v'$ have $i$ and $i'$ numbered
neighbors respectively, with $i'<i$, and furthermore if the $i'$ first
numbered neighbors of $v$ are exactly the last $i$ numbered neighbors
of $v'$ in the same order, then $v$ is also preferred.

Note that this intuition is the same than for LexUP, apart that the
word first have been replaced by the word last. Indeed the same
integers is added to the label in both cases. However, in the former
case the integer is prepended while in the latter case the integer is
appended. Hence, in both cases, neighbors with small number are
prefered. But in LexUP they must be the earliest neighbors while in
LexDFS they must be the latest neighbors.

Figure \ref{fig:DFS} show an example of a LexBFS ordering.  Each arrow
associates to a vertex $v$ its last numbered neigbhor.  Table
\ref{tab:DFS} associates to each vertex $v$ its label when $v$ was
numbered. Note that this list is also its list of numbered neighbors.
\begin{table}[h]
  \centering
  \begin{tabular}{l||l|l|l|l|l|l|l|l}
    Vertex number&2&3&4 &5 &6 &7 &8 &9 \\
    \hline
     $v$'s label  in Figure \ref{fig:DFS}  &1&2&31&4&53&6 &73 &82 \\
  \end{tabular}
\caption{Label during the LexDFS search}
\label{tab:DFS}
\end{table}

\subsection{LexDown}\label{sec:def-down}
Let \index{Lexicographic DOWN}Lexicographic DOWN\index{LexDOWN} be a
graph search algorithm, as in \autoref{def:search}, where Line
\ref{search-def-loop-neighb} is: \quo{prepend $n-i$ to the label of
\var{\var{neighb}}}.

Intuitively, at each step, the vertex $v$ is preferred to the vertex
$v'$ if the first numbered neighboor of $v$ have been numbered later
than the first numbered neighboor of $v'$.  If their first neighboor
are equal, then the same comparaison is done on the second
neighbor. And so on. If $v$ and $v'$ have $i$ and $i'$ numbered
neighbors respectively, with $i'<i$, and furthermore if the $i'$ first
numbered neighbors of $v$ are exactly the last $i$ numbered neighbors of
$v'$ in the same order, then $v$ is also preferred.

Note that this intuition is the same than for LexBFS (respectively,
LexDFS), apart that the word earlier (respectively, first) have been
replaced by the word later (respectively, last). The reason is similar
to the previous explanations.

\autoref{fig:DOWN} shows an examples of LexBFS ordering.  Each arrow
associates to a vertex $v$ its last numbered neigbhor.  Table
\ref{tab:DOWN} associates to each vertex $v$ its list of numbered
neighbors when $v$ was numbered. This table also associate to $v$ its
label.
\begin{table}[h]
  \centering
  \begin{tabular}{l||l|l|l|l|l|l|l|l}
    Vertex number &2&3&4&5&6&7&8&9\\
    \hline
    $v$'s label  in Figure \ref{fig:DOWN}             &8&8&7&67&6&45&2&134\\
    \hline
    When $v$ is numbered,  its numbered neighbors are &1&1&2&32&3&54&7&865
\end{tabular}
\caption{Labels during the LexDOWN search}
\label{tab:DOWN}
\end{table}
Note that when vertices 1 and 2 are fixed, this graph admits no other
LexBFS ordering.
 \section{LexDFS and LexDOWN}\label{sec:DFS-DOWN}

An algorithm is now given, which outputs a LexDOWN ordering in time
$\bigO{n+m\log(m)}$. Note that $m\le n^{2}$, hence
$\log(m)\le 2\log n$, thus, this algorithm is more efficient than
\autoref{def:search}. 
\begin{theorem}\label{theo:DFS}
  Let $G=(V,E,s)$ be a connex undirected \type{graph} with source $s$, with
  $n$ vertices and $m$ edges. A LexDFS ordering of $G$ can be computed
  in time $\bigO{n+m\log(m)}$.
\end{theorem}
An intuition of the algorithm is first given. Note that, in
\autoref{def:search}, at each iteration of the loop of Line
\ref{search:choose}, all unnumbered states must be checked. At each
iteration, this line  runs in time  $\bigO{m}$. This time can be avoided if
the list is already sorted. Since at the $i$-th iteration, at most
$\card{N(v_{\sigma(i)})}$ labels change, it suffices to sort and move
those $\bigO{\card{N(v_{\sigma(i)})}}$ elements. The sorting can be
done in time
$\bigO{\card{N(v_{\sigma(i)})}\log(\card{N(v_{\sigma(i)})})}$. Since
all of those elements must be moved to the front of the \type{list}, a
correct usage of pointers allow to move the
$\bigO{\card{N(v_{\sigma(i)})}}$ vertices in time
$\bigO{\card{N(v_{\sigma(i)})}}$.  Summing over all $i$, the times
taken by those operationsày is
$\bigO{\sum_{i=1}^{n}\card{N(v_{\sigma(i)})}\log(\card{N(v_{\sigma(i)})})=m\log(m)}$.

A simplified version of the algorithm is given as
\autoref{alg:lexDFS-Simple}. In this simplified version, \type{vertex}
is a type which contains an \type{integer} \param{order} and a
\param{label}.  \autoref{alg:lexDFS} furthermore shows exactly how to
use pointers in order to obtain a quasi-linear time.
 \begin{algorithm}[!h]\caption{Computing a LexDFS ordering-simplified}\label{alg:lexDFS-Simple}
    \KwIn{$G=(V,E,s)$ an undirected \type{graph} with a source}
    \KwOut{a LexDFS-Simple ordering $\sigma$ of the vertices of $G$}
    \lnl{DFS-Simple-init-sigma}$\sigma$: \type{array} of $n$ \type{integer}s\;
    \lnl{DFS-Simple-vertices-decl}vertices: \type{array} of $n$ elements of type \type{vertex}\;
    \lnl{DFS-Simple-init-max}\var{max}:=0\;
    \lnl{DFS-Simple-init-loop}\ForEach(\tcc*[f]{Initialization}){$i$ from 1 to $n$}{
      \lnl{DFS-Simple-init}vertices[$i$]:=\{\param{order}:=$-\infty$
      \param{label}:=[]\}\;
    }
    \lnl{init-pos-s}vertices[$s$]:=\{\param{order}:=0; \param{label}:=[$\infty$]\}\; 
    \lnl{DFS-Simple-unnumbered-init}\var{unnumbered}:=[$s$]\;
    \lnl{DFS-Simple-main-loop}\ForEach{$i$ from 1 to $n$}{
      \lnl{DFS-Simple-sigma-assign}$\sigma(i):=$\var{unnumbered}.\param{first}.\param{value}\tcc*{Selecting
      the greatest value.}
      \lnl{DFS-Simple-unnumbered-remove}remove $\sigma(i)$ from \var{unnumbered}\;
      \lnl{DFS-Simple-sort}sort the neighbors of $v_{\sigma(i)}$ in increasing order\;
       \lnl{DFS-Simple-inside-loop}\ForEach{\var{neighb}: \var{unnumbered} neighbor of $v_{\sigma(i)}$ in increasing order}{
         \lnl{DFS-Simple-if-remove}\If(\tcc*[f]{\var{neighb} must be removed from }){neighb's \param{label} is empty}{
          \lnl{DFS-Simple-remove} remove \var{neighb} from \var{unnumbered}\tcc*{\var{unnumbered} if its was in it.}
         }
        \lnl{DFS-Simple-prepend-label} prepend $i$ to neighb's \param{label}\tcc*{\var{neighb} now has the greatest label}
        \lnl{DFS-simple-last-neighb}add \var{neighb} to the front of \var{unnumbered}\;
        \lnl{DFS-Simple-incr-max} set \var{max} to \var{max}+1\;
        \lnl{DFS-Simple-max}set neighb's \param{order} to \var{max}\tcc*{and has the greatest \param{order}}
       }
    }
    \KwRet{$\sigma$}
   \end{algorithm}
\begin{proof}
  In \autoref{alg:lexDFS}, a \type{vertex} is a data structure which
  contains 4 parameters
  \begin{itemize}
  \item \param{order} : an \type{integer};
  \item \param{pos} : a \type{node} of a \type{list} of \type{integers};
  \item \param{label} : a \type{list} of \type{integer}s;
  \item \param{numbered} : a \type{Boolean};
  \end{itemize}

  \begin{algorithm}[!h]\caption{Computing a LexDFS ordering}\label{alg:lexDFS}
    \KwIn{$G=(V,E,s)$ an undirected \type{graph} with a source}
    \KwOut{a LexDFS ordering $\sigma$ of the vertices of $G$}
    \SetKwRepeat{Struct}{struct \{}{\}}
    \lnl{DFS-init-sigma}$\sigma$: \type{array} of $n$ \type{integer}s initialized to $-1$\;
    \lnl{DFS-vertices-decl}vertices: \type{array} of $n$ elements of type \type{vertex}\;
    \lnl{DFS-init-max}\var{max}:=0\;
    \lnl{DFS-init-loop}\ForEach(\tcc*[f]{Initialization}){$i$ from 1 to $n$}{
      \lnl{DFS-init}vertices[$i$]:=\{\param{order}:=$-\infty$;
      \param{label}:=[];
      \param{numbered}:=false\}\;
    }
    \lnl{DFS-unnumbered-init}\var{unnumbered}:=[$s$]\;
    \lnl{init-pos-s}vertices[$s$]:=\{order:=0; \param{pos}:=\var{unnumbered}.\param{last};
    label    :=[$\infty$]; numbered:=false\}\; 
    \lnl{DFS-main-loop}\ForEach{$i$ from 1 to $n$}{
      \lnl{DFS-sigma-assign}$\sigma(i):=$\var{unnumbered}.\param{first}.\param{value}\tcc*{Selecting
      the greatest value.}
      \lnl{DFS-unnumbered-remove}\var{unnumbered}.\param{first}.\param{remove}\;
      \lnl{DFS-numbered}vertices[$\sigma(i)$].\param{numbered}:=true\;
      \lnl{DFS-sort}sort the neighbors of $v_{\sigma(i)}$ in increasing order\;
      \lnl{DFS-inside-loop}\ForEach{\var{neighb}: \var{unnumbered} neighbor of $v_{\sigma(i)}$ in increasing order}{
        \lnl{DFS-if-remove}\If(\tcc*[f]{\var{neighb} must be removed from }){vertices[\var{neighb}].\param{label}$\ne[]$}{
          \lnl{DFS-remove}vertices[\var{neighb}].\param{pos}.\param{remove}\tcc*{\var{unnumbered} if its was in it.}
        }
        \lnl{DFS-prepend-label}vertices[\var{neighb}].\param{label}.\param{add-first}.($i$)\tcc*{\var{neighb} now has the greatest label}
        \lnl{dfs-last-neighb}\var{unnumbered}.\param{add-first}.(\var{neighb})\tcc*{hence, it goes in front of the list,}
        \lnl{DFS-pos-neighb}vertices[\var{neighb}].\param{pos}:= \var{unnumbered}.\param{first}\;
        \lnl{DFS-incr-max}\var{max}:=\var{max}+1\;
        \lnl{DFS-max}vertices[\var{neighb}].\param{order}:=\var{max}\tcc*{and has the greatest \param{order}}
      }
    }
    \KwRet{$\sigma$}
  \end{algorithm}
  Let us first prove that \autoref{alg:lexDFS} returns a lexDFS
  ordering.  For $1\le i\le n$, let $\sigma_{i}$, \var{unnumbered}$_{j}$,
  vertices$_j$ and \var{max}$_j$ be the values of those variables when the
  iteration of the loop of Line \ref{DFS-main-loop} ends, with the
  variable $i$ interpreted by $j$. Finally, let $\sigma_{0}$,
  \var{unnumbered}$_{0}$, vertices$_0$ and \var{max}$_0$ be the values of those
  variables before the first iteration of this loop.
  
  The loop invariants of this algorithm are:
  \begin{enumerate}
  \item\label{DFS-sigma} $\sigma_{j}[i]$ contains an
    element $k$ such that label$_{i}(v_{k})$ is lexicographically
    maximal, for $0<i\le j$.
  \item\label{DFS-label} vertices$_{j}[i]$.\param{label} contains the
    label of $v_{i}$, as in the $j$-th step of LexDFS.  Note that
    vertices$[i]$.\param{label} is not actually used in computation
    of the LexDFS ordering.
  \item\label{DFS-unnumbered} The variable \var{unnumbered}$_{j}$ contains
    the \type{list} of \var{unnumbered} vertices with a non-empty label. Those
    vertices appears in decreasing lexicographic order of their
    labels.
  \item\label{DFS-order-bound} If $x$ appears before $y$ in
    \var{unnumbered}$_{j}$, then
    $node_{j}[x].\param{order}>node_{j}[y].\param{order}$.  
  \item\label{DFS-pos} If \var{unnumbered}$_{j}$ contains the vertex
    $v_{i}$, then vertices$_{j}[i]$.\param{pos} is the \type{node}
    of \var{unnumbered}$_{j}$ whose value is $v_{i}$. Otherwise,
    vertices$_{j}[i]$.\param{pos} is unspecified.
  \item\label{DFS-max-min}\var{max}$_{j}$ is greater than all finite
    vertices$_j$[$i$].\param{order}.
  \item\label{DFS-max-bound}\var{max}$_{j}$ is less than the sum of the
    degree of the vertices $v$ which are numbered at the $j$-th step.
    \label{DFS-last}
  \end{enumerate}
  Let us show that, for $0\le j\le n$, the \ref{DFS-last} invariant
  are satisfied.  Invariants \ref{DFS-pos} is satisfied at each step,
  because everytime an \type{integer} $i$ is added into \var{unnumbered},
  vertices[$i$].\param{pos} is modified accordingly. The proof for
  the other invariants is by induction on $j$.

  \paragraph{}
  Let us show that, for $j=0$, the \ref{DFS-last} invariant are
  satisfied.

  Invariant \ref{DFS-sigma} holds, since there are no {integer}
  $0<i\le 0$.

  By definition of LexDFS, all labels are empty at
  initialization, apart from the one of the source. It is the case in
  this program because of Lines \ref{DFS-init} and
  \ref{init-pos-s}. Hence invariant \ref{DFS-label} is satisfied.

  Note that $v_{s}$ is the only labelled vertex and that no vertex is
  numbered.  Furthermore $s$ is the element of \var{unnumbered} because of
  Line \ref{DFS-unnumbered-init}. Hence invariant \ref{DFS-unnumbered}
  is satisfied.

  Invariant \ref{DFS-order-bound} is also trivially satisfied, since
  $s$ have the greatest order and the greates label, and all other
  orders are equal and all other label are equals.

  Invariant \ref{DFS-max-min} is trivially satisfied since
  for all $i$, vertices$_0$[$i$].\param{order}=0.

  Invariant
  \ref{DFS-max-bound} is trivially satisfied since no vertices are
  numbered at the 0-th step.

\paragraph{}

  Let $0< j\le n$. Let us now assume that the \ref{DFS-last}
  invariants holds at step $j-1$, and let us prove that it holds for
  $j$.

  Since Invariant \ref{DFS-sigma} holds at step $j-1$, it clearly
  holds at step $j$ for all $i< j$. It remains to consider the case
  $i=j$. By invariant \ref{DFS-unnumbered}, \var{unnumbered}$_j$ contains
  the \type{list} of \var{unnumbered} labelled vertices at step $j$, in decreasing
  lexicographic order of their labels. Hence Line
  \ref{DFS-sigma-declare} correctly assigns to $\sigma[i]$ a vertex
  $w$ such that label$_{j}(w)$ has a maximal label. Thus, Invariant
  \ref{DFS-sigma} holds at step $j$.

  Invariant \ref{DFS-label} clearly remains true since the updating of
  the label is exactly the one of the definition of the LexDFS
  algorithm.
  
  At the $j$-th step, the list of \var{unnumbered} vertices with a non-empty
  label contains, in this order:
  \begin{itemize}
  \item The neighbors of $v_{\sigma(j)}$, which are \var{unnumbered} and
    have a non-empty label at step $j-1$. The order, according to
    their labels, are in the same order in both lists.
  \item The vertices which are neither $v_{\sigma(j)}$ nor its
    neighbors, which are \var{unnumbered} and have a non-empty label at step
    $j-1$. The order, according to their labels, are in the same order
    in both \type{list}s.
  \item The neighbors of $v_{\sigma(j)}$ which are \var{unnumbered} and have
    an empty label at step $j-1$. 
  \end{itemize}
  Thus, according to invariant \ref{DFS-unnumbered}, \var{unnumbered}$_{j}$
  must contains, in the following order:
  \begin{itemize}
  \item the elements of \var{unnumbered}$_{j-1}$  which are neighbors of
    $v_{\sigma(j)}$, in the same order,
  \item the elements of \var{unnumbered}$_{j-1}$ which are neither neighbors
    of $v_{\sigma(j)}$ nor $j$, in the same order,
  \item the unlabelled neighbors of $v_{\sigma(j)}$, in an arbitrary
    order.
  \end{itemize}
  This is indeed the value of \var{unnumbered}$_{j}$, because of Lines
  \ref{DFS-remove} and \ref{dfs-last-neighb}.  Hence invariant
  \ref{DFS-unnumbered} holds at step $j$.

  Since each time an element is moved to the front of
  \var{unnumbered}, its \param{order} is greater than any
  \param{order} presently assigned \ref{DFS-pos}, its \param{order}
  is greater than any previously assigned \param{order}, then
  Invariant \ref{DFS-order-bound} holds.

  Invariant \ref{DFS-max-min} clearly holds since the \param{order}s
  are assigned in increasing order, and since, each time an
  \param{order} is assigned, \var{max} is assigned to be its
  predecessor.

  It is easy to see that
  $\var{max}_j\le \var{max}_{j-1}+\card{N(\sigma(j))}$. Hence
  invariant \ref{DFS-max-bound} is true at step $j$.
  
  Since the invariants are satisfied at each steps, by
  \ref{DFS-sigma}, at the end of the loop, $\sigma$ contains a
  LexDFS ordering of $G$. Hence the algorithm indeed returns a
  LexDFS ordering of $G$.
  \paragraph{}
  Let us now consider the computation time.  The code of Lines
  \ref{DFS-init-sigma}, \ref{DFS-vertices-decl}, \ref{DFS-init-max},
  \ref{init-pos-s} and \ref{DFS-unnumbered-init} are executed exactly
  once, and runs in time $\bigO{n}$. Hence their cost is $\bigO{n}$.

  Lines \ref{DFS-init}, \ref{DFS-sigma-assign},
  \ref{DFS-unnumbered-remove},  \ref{DFS-numbered}, are executed $n$
  times and runs in constant time. Hence their cost is $\bigO{n}$.
  
  Line \ref{DFS-sort} is executed once for each vertex $v_{i}$. And
  for each vertex $v_{i}$, it runs in time
  $\bigO{\card{N(v_{i})}\log(\card{N(v_{i})})}$. Hence the total cost
  of this line is
  $\bigO{\sum_{i=1}^{n}\card{N(v_{i})}\log(\card{N(v_{i})})}=\bigO{m\log
    m}$.

  Lines \ref{DFS-prepend-label} to \ref{DFS-max} are executed once by
  edge, and executed in constant time. Hence their cost is $\bigO{m}$.

  Finally, the total execution time is $\bigO{n+m\log(m)}$.
\end{proof}
Note that the \param{order}s are either infinite, or \type{integer}s between
0 and $2m$. Hence it is acceptable to assume that comparison of two
\param{order} parameters can be done in constant time.

\paragraph{LexDOWN}
As stated in \autoref{sec:def-down}, LexDOWN is similar to LexDFS. It
is now considered.
\begin{theorem}
  Let $G=(V,E,s)$ be a connex undirected \type{graph} with source $s$, with
  $n$ vertices and $m$ edges. A LexDOWN ordering of $G$ can be
  computed in time $\bigO{n+m\log(m)}$.
\end{theorem}

\begin{proof}
  The algorithm to compute a LexDOWN ordering is \autoref{alg:lexDFS},
  with the three following changes:
  \begin{itemize}
  \item Line \ref{DFS-prepend-label} is tranformed into
    \quo{vertices[\var{neighb}].\param{label}.\param{prepend}.($n- i$)},
  \item Line \ref{dfs-last-neighb} is transformed into
    \quo{\var{unnumbered}.\param{add-last}.(\var{neighb});} and
  \item Line \ref{DFS-incr-max} is transformed into \quo{\var{max}:=\var{max}-1;}.
  \end{itemize}
  Invariant \ref{DFS-max-min} must be changed to \quo{\var{max}$_{j}$ is
    smaller than all finite vertices$_j$[$i$].\param{order}}, and
  \ref{DFS-max-bound} must be changed to \quo{|\var{max}$_{j}$| is less than
    the sum of the degree of the vertices $v$ which are numbered at
    the $j$-th step}.  Apart from those changes, the proof of this
  theorem is exactly the same than the proof of \autoref{theo:DFS}.
\end{proof}
\section{Efficient LexBFS and LexUP}\label{sec:BFS-UP}
In this section, it is shown that a LexUP ordering can be computed in
linear time. The algorithm is very similar to the algorithm for
efficiently computing LexBFS.

A simplified version of the linear time algorithm which computes a
LexBFS ordering is recalled as \autoref{alg:LexBFS-simple}.  This
algorithm keeps a \type{list}, \var{unnumbered}, which contains all
vertices, with a non-empty label, in decreasing order according to
their label. More precisely, all (indices of) vertices with the same
non-empty label belong to a \type{set}, and \var{unnumbered} is a
\type{list} of \type{set}s. The \type{set}s are also encoded as \type{list}s. When a
vertex $v_{i}$ is numbered, the label of its neighbors
increases. However, it does not increase enough to become greater than
labels which used to be greater than it.  Hence all neighbors
belonging to the same \type{set} $s$ are moved to a new \type{set}
$s'$ placed before $s$. As soon as a \type{set} is empty, it is
removed from the list. Each vertex $v$ is moved at most $\card{N(v)}$
times in the list.

A correct usage of pointers, as shown in \autoref{alg:LexBFS}, allows
to move indices from the previous \type{set} to the new \type{set} in
constant time.  Hence, the algorithm runs in time $\bigO{n+m}$.  In
this algorithm, a \type{set} is a data-structure with three
parameters:
\begin{itemize}
\item  \param{pos}: a \type{node} of a \type{list} of \type{set}s,
\item  \param{edited}: an \type{integer} and
\item  \param{elements}: a \type{list} of \type{integer}s.
\end{itemize}
And a \type{vertex} is a data-structure with four parameters:
\begin{itemize}
\item \param{pos}: a \type{node} of a \type{list} of \var{integer}s,
\item \param{numbered}: a \type{Boolean},
\item \param{label}: a \type{list} of \type{integer}s and
\item \param{set}: a \type{set}.
\end{itemize}

\begin{algorithm}[!h]
  \caption{Efficient computation of a LexBFS - simplified}\label{alg:LexBFS-simple}
  \KwIn{An undirected \type{graph} $G=(V,E)$}
  \KwOut{an ordering $\sigma$ of the vertices of $G$}
  \SetKwRepeat{Struct}{struct \{}{\}}
  vertices: \type{array} of $n$ elements of type \type{vertex}\;
  \lnl{DFS-loop-init}\ForEach(\tcc*[f]{Initialization}){$i$ from 1 to $n$, distinct from $s$}{
    $v_{i}$'s \param{label} is set to []\;
  }

  set\_s is set to [$s$]\;
  $s$'s label is set to $[\infty]$\;

  \var{unnumbered}:=[set\_s]\;

  \ForEach{$i$ from $1$ to $n$}{
    \var{greatest\_set}:= the first element of \var{unnumbered}\;
    \lnl{DFS-sigma-declare}$\sigma[i]$:= any element of \var{greatest\_set}\tcc*{Selecting  a greatest vertex}
    Remove this element from \var{greatest\_set}\tcc*{and  removing it from the list.}
    If \var{greatest\_set} is empty, remove it from \var{unnumbered}\;
    \ForEach{$\text{\var{neighb}}\in N(v_{\sigma(i)})$, \var{unnumbered}}{
      \eIf{\var{neighb}'s \param{label} is not empty}{
        \eIf{no vertices from \var{neighb}'s \param{set} have been seen for this value of $i$}{
          new\_set is set to the []\;
          add \var{new\_set} before \var{neighb}'s \param{set}\;
        }{
          set \var{new\_set} to the set preceding \var{neighb}'s \param{set}\;
        }
        If \var{neighb}'s \param{set} is a singleton, remove this set from \var{unnumbered}\;
      }(\tcc*[f]{If \var{neighb}'s label is empty}){
        \If{No ununlabelled neighbor have been seen for this value of $i$}{
          Set \var{new\_set} to a new set\;
          Add \var{new\_set} to the rear of \var{unnumbered}\;
        }{
          Set \var{new\_set} to the last set of \var{unnumbered}\;
        }
      }
      Move \var{neighb} to \var{new\_set}\;
      \lnl{BFS-simple-change-label} append $n-i$ to \var{neighb}'s \param{label}
    }
  }
  \KwRet{$\sigma$}\;  
\end{algorithm}
\begin{algorithm}[!h]
  \caption{Efficient computation of a LexUP}\label{alg:LexBFS}
  \KwIn{An undirected \type{graph} $G=(V,E)$}
  \KwOut{an ordering $\sigma$ of the vertices of $G$}
  \SetKwRepeat{Struct}{struct \{}{\}}
  vertices: \type{array} of $n$ elements of type \type{vertex}\;
  unlabelled\_edited:=0\;
  \lnl{DFS-loop-init}\ForEach(\tcc*[f]{Initialization}){$i$ from 1 to $n$, distinct from $s$}{
    vertices[$i$]:=\{\param{numbered}:=false, \param{label}:=[]\}\;
  }

  set\_s:=\{\param{pos}:=set\_s;\param{edited}:=0;\param{elements}:=[$s$]\}\;
  vertices[$s$]:=\{\param{pos}:=set\_s.\param{elements}.\param{first},\param{numbered}:=false,
  \param{label}:=[$\infty$];\type{set}:=set\_s\}\;

  \var{unnumbered}:=[set\_s]\;

  \ForEach{$i$ from $1$ to $n$}{
    \var{greatest\_set}:=\var{unnumbered}.\param{first}.\param{value}\;
    \lnl{DFS-sigma-declare}$\sigma[i]:=$\var{greatest\_set}.\param{elements}.\param{first}.\param{value}\tcc*{Selecting
      a greatest vertex}
    \var{greatest\_set}.\param{elements}.\param{first}.\param{remove}\tcc*{and  removing it from the list.}
    vertices[$\sigma(i)$].\param{numbered}:=true\;
    \If{\var{greatest\_set}.\param{elements}=[]}{
      \var{greatest\_set}.\param{pos}.\param{remove}\;
    }
    \ForEach{$\text{\var{neighb}}\in N(v_{\sigma(i)})$, \var{unnumbered}}{
\eIf(\tcc*[f]{If the neighbor's label is not empty}){vertices[\var{neighb}].\param{label}$\ne[]$}{
        \eIf(\tcc*[f]{no neighbors with the same label have been
          seen: a new set must be created before the current one.}){vertices[\var{neighb}].\type{set}.\param{edited}<$i$}{
          vertices[\var{neighb}].\type{set}.\param{edited}:=$i$\;
          \var{new\_set}:=\{\param{edited}:=$i$; \param{elements}:=[]\}\;
          \lnl{DFS-create-set}vertices[\var{neighb}].\type{set}.\param{pos}.\param{add-before}(\var{new\_set})\;
          \lnl{DFS-new-set-before}\var{new\_set}.\param{pos}:= vertices[\var{neighb}].\type{set}.\param{prec}\;
        }(\tcc*[f]{A neighbor with the same label have already been seen}){
          \var{new\_set}:=vertices[\var{neighb}].\type{set}.\param{pos}.\param{prec}\;
        }
        vertices[\var{neighb}].\param{pos}.\param{remove}\;
        \If{vertices[\var{neighb}].\type{set}.\param{elements}=[]}{
          vertices[\var{neighb}].\type{set}.\param{remove}\;
        }
      }(\tcc*[f]{If \var{neighb}'s label is empty}){
        \eIf(\tcc*[f]{No unlabelled neighboors have been considered yet.}){unlabelled\_edited<$i$}{
          unlabelled\_edited:=$i$\;
          \var{new\_set}:=\{\param{edited}:=$i$; \param{elements}:=[]\}\;
          \lnl{BFS:new-set-unlabelled}\var{unnumbered}.\param{add-last}(\var{new\_set})\;
          \var{new\_set}.\param{pos}:= \var{unnumbered}.\param{last}\;
        }{
          \var{new\_set}:= \var{unnumbered}.\param{last}.\param{value}\;
        }
      }
      \var{new\_set}.\param{elements}.\param{add-last}(\var{neighb})\tcc*{Moving
      the neighbor}
      vertices[\var{neighb}].\type{set}:=\var{new\_set}\;
      vertices[\var{neighb}].\param{pos}:=\var{new\_set}.\param{elements}.\param{last}\;
      \lnl{BFS-change-label}vertices[\var{neighb}].\param{label}.\param{add-last}($n-i$)\tcc*{Updating
        the label of  \var{neighb}}
    }
  }
  \KwRet{$\sigma$}\;  
\end{algorithm}

Note that if a vertex $v$ have an empty label, it has the
lexicographically smallest label. Hence, when a first element is added
to the label of $v$, this vertex moves to the second least set (which
may become the least set if there remains no more vertex with an empty
label). Indeed, the first element of $v$'s label is $n-i$. And $n-i$
is smaller than the first element of the label of all other vertices
$w$ with non-empty label.

The preceding remark leads to the main difference between LexBFS and
LexUP.  In LexUP, the element added is $i$, and not $n-i$. Hence the
first element of this vertex $v$ is $i$. Hence, it is greatest than
all the first element of all other vertices $w$ with a non-empty
label. Hence, $v$ moves to the greatest set. Therefore, to transform
\autoref{alg:LexBFS} 
into an algorithm which computes a  LexUP ordering, it suffices to do
the following change:
\begin{itemize}
\item  Line \ref{BFS:new-set-unlabelled} must be modified to
  \quo{\var{unnumbered}.\param{add-first}(\var{new\_set});}. 
\item Line \ref{BFS-change-label} must be changed to
  \quo{vertices[\var{neighb}].\param{label}.\param{add-last}($i$);}.
\end{itemize}

The proof that \autoref{alg:lexDFS} computes a LexDFS ordering is
similar to the proof that \autoref{alg:LexBFS} computes a LexBFS
ordering.

Note that, if \var{unlabelled} was not restricted to contains only
labelled vertices, the algorithm would still be correct for
LexBFS. Furthermore, the algorithm would be be shorter. However, the
algorithm will not be correct anymore for LexUP.

\section{Conclusion}
In this paper, it has been proven that a LexUP ordering can be
computed in linear time and that a LexDOWN ordering and a LexDFS
ordering can be computed in time $\bigO{n+m\log m}$.

The author thanks Michel Habib, who introduced this problem to him
during his Graph Theory Lectures.

\bibliographystyle{alpha}
\bibliography{../fo}

\printindex{}

\end{document}

%% file: preambule.tex
\usepackage[utf8]{inputenc}
\usepackage[T1]{fontenc}
\usepackage[nottoc,notlof, notlot,section]{tocbibind} 

\usepackage{xparse}

\usepackage{appendix}
\usepackage{afterpage}
\usepackage{multicol}
\usepackage{color}
\usepackage[usenames]{xcolor}
\usepackage{xspace}
\usepackage{setspace}
\usepackage{multido}
\usepackage[normalem]{ulem}
\usepackage{soul}
\usepackage{url}

\usepackage{color}
\usepackage{csquotes}
\usepackage[inline]{enumitem}
\usepackage{wrapfig}
\usepackage{framed}
\usepackage[hypcap]{caption}
\usepackage{subcaption}
\usepackage{tikz}
\usetikzlibrary{positioning,shadows,arrows,automata}
\usepackage{graphics}
\usepackage{graphicx}

\usepackage{multirow}
\usepackage{makecell}
\usepackage{etoolbox}
\usepackage{times}
\usepackage{mathtools}
\usepackage{amsmath}
\usepackage{mathtools}
\usepackage{amssymb} 
\usepackage{amsthm}
\usepackage{centernot}

\usepackage[cyr]{aeguill}
\usepackage{tipa}
\usepackage{stmaryrd} 

\usepackage[boxruled,linesnumbered,noend]{algorithm2e}
\usepackage{algorithmic}
\usepackage{complexity}
\newclass{\MSO}{MSO}
\newclass{\WMSO}{WMSO}
\newclass{\SO}{SO}
\usepackage{latexsym} 
\usepackage{gensymb}

\newcommand{\ignore}[1]{}
\let\originalleft\left
\let\originalright\right
\renewcommand{\left}{\mathopen{}\mathclose\bgroup\originalleft}
\renewcommand{\right}{\aftergroup\egroup\originalright}

\newcommand{\mb}{\mathbb}

\newcommand{\bigO}[1]{O\left(#1\right)}

\DeclareDocumentCommand{\wordToReal}{ O{b}  m}{\wordToNumber[#1]{#2}^{\R}}
\DeclareDocumentCommand{\wordToFractional}{ O{b}  m}{\wordToNumber[#1]{#2}^{\fra}}
\DeclareDocumentCommand{\wordToNatural}{ O{b}  m}{\wordToNumber[#1]{#2}^{\nat}}
\DeclareDocumentCommand{\wordToNumber}{ O{b}  m}{
  \left[{#2}\right]_{#1}
}
\DeclareDocumentCommand{\wordToTupleReal}{ O{b} O{d} m}{\wordToNumber[#1][#2]{#3}^{\R}}
\DeclareDocumentCommand{\wordToTupleFractional}{ O{b} O{d} m}{\wordToNumber[#1][#2]{#3}^{\fra}}
\DeclareDocumentCommand{\wordToTupleNatural}{ O{b} O{d} m}{\wordToNumber[#1][#2]{#3}^{\nat}}
\DeclareDocumentCommand{\wordToTupleNumber}{ O{b} O{d} m}{
  \def\temp{#2}\ifx\temp\empty
  \left[#3\right]_{#1}
  \else
  \left[#3\right]_{#1,#2}
  \fi
}
\DeclareDocumentCommand{\logicToTupleNumber}{ O{d} m}{
  \left[#2\right]_{#1}
}
\DeclareDocumentCommand{\autToTupleNumber}{ O{b} O{d} m}{
  \def\temp{#2}\ifx\temp\empty
  \left[#3\right]_{#1}
  \else
  \left[#3\right]_{#1,#2}
  \fi
}

\newcommand{\realDot}{\star}

\newcommand{\set}[1]{\left\{ #1\right\}}

\newcommand{\N}{\mathbb N{}}
\newcommand{\nat}{I}

\DeclareDocumentCommand{\digitSet}{ O{b}}{\Sigma_{#1}}
\DeclareDocumentCommand{\digitSetDim}{ O{b} O{d}}{\Sigma_{#1,#2}}
\DeclareDocumentCommand{\digitDotSetDim}{ O{b} O{d}}{\Sigma_{#1,#2}\cup\set{\realDot}}

\renewcommand{\R}{\ensuremath{\mb{R}}}

\newcommand{\card}[1]{|#1|}

\DeclareMathOperator{\fra}{frac}

\usepackage{xstring}
\newcommand{\appendExp}[1]{
  \hspace{.5cm}
  \left(#1\right)
}
\newcommand{\anot}[3]{
  &#3&#2&
  \def\temp{#1}\ifx\temp\empty
  \else\appendExp{#1}
  \fi
}

\usepackage{aliascnt}
\usepackage{hyperref}

%% file: en-preambule.tex
\newtheorem{theorem}{Theorem}[section]

\newaliascnt{lemma}{theorem}  
\newtheorem{lemma}[lemma]{Lemma}  
\aliascntresetthe{lemma}

\newaliascnt{corollary}{theorem}  
  
\aliascntresetthe{corollary}

\theoremstyle{definition}

\newaliascnt{example}{theorem}  
  
\aliascntresetthe{example}

\newaliascnt{definition}{theorem}  
  
\aliascntresetthe{definition}

\newaliascnt{notation}{theorem}  
  
\aliascntresetthe{notation}

\newaliascnt{proposition}{theorem}  
  
\aliascntresetthe{proposition}

\newaliascnt{property}{theorem}  
  
\aliascntresetthe{property}

\theoremstyle{remark}
\newaliascnt{remark}{theorem}  
  
\aliascntresetthe{remark}

\newcommand{\quo}[1]{\text{``}#1\text{''}}




%% file: graphe.bbl
\begin{thebibliography}{CSRL01}

\bibitem[CK08]{DBLP:journals/siamdm/CorneilK08}
Derek~G. Corneil and Richard Krueger.
\newblock A unified view of graph searching.
\newblock {\em {SIAM} J. Discrete Math.}, 22(4):1259--1276, 2008.

\bibitem[CSRL01]{Cormen}
Thomas~H. Cormen, Clifford Stein, Ronald~L. Rivest, and Charles~E. Leiserson.
\newblock {\em Introduction to Algorithms}.
\newblock McGraw-Hill Higher Education, 2nd edition, 2001.

\bibitem[Dus14]{dusart:tel-01273352}
J{\'e}r{\'e}mie Dusart.
\newblock {\em {Graph Searches with Applications to Cocomparability Graphs}}.
\newblock Theses, {Universit{\'e} Denis Diderot Paris 7}, June 2014.

\bibitem[RTL76]{DBLP:journals/siamcomp/RoseTL76}
Donald~J. Rose, Robert~Endre Tarjan, and George~S. Lueker.
\newblock Algorithmic aspects of vertex elimination on graphs.
\newblock {\em {SIAM} J. Comput.}, 5(2):266--283, 1976.

\end{thebibliography}
